\documentclass[lettersize,journal]{IEEEtran}
\usepackage{amsmath,amsfonts}
\usepackage{algorithmic}
\usepackage{algorithm}
\usepackage{array}
\usepackage[caption=false,font=normalsize,labelfont=sf,textfont=sf]{subfig}
\usepackage{textcomp}
\usepackage{stfloats}
\usepackage{url}
\usepackage{verbatim}
\usepackage{graphicx}
\usepackage{cite}
\usepackage{amsthm}
\usepackage{amssymb}
\usepackage{enumerate}
\usepackage{subfig}
\usepackage{setspace}
\usepackage{bm}
\usepackage{textcomp}
\usepackage{xcolor}
\usepackage[scr=boondoxo,scrscaled=1.05]{mathalfa}
\usepackage{upgreek}
\usepackage{marginnote}

\newcommand{\R}{{\mathbb R}}

\newcommand{\N}{{\mathbb N}}

\newcommand{\J}{{\cal J}}

\newcommand{\be}{\begin{equation}}
\newcommand{\ee}{\end{equation}}
\newcommand{\ba}{\begin{array}}
\newcommand{\ea}{\end{array}}
\newcommand{\baa}{\left[\begin{array}}
\newcommand{\eaa}{\end{array}\right]}

\newcommand{\beqa}{\begin{eqnarray}}
\newcommand{\eeqa}{\end{eqnarray}}
\newcommand{\bt}{\begin{tabular}}
\newcommand{\et}{\end{tabular}}

\newcommand{\bi}{\begin{itemize}}
\newcommand{\ei}{\end{itemize}}
\newcommand{\bc}{\begin{center}}
\newcommand{\ec}{\end{center}}

\newtheorem{lemma}{Lemma}
\newtheorem{prop}{Proposition}
\newtheorem{remark}{Remark}

\newcommand{\abs}[1]{\left|#1\right|}
\newcommand{\norm}[1]{\left\lVert#1\right\rVert}

\newcommand{\eor}{\ensuremath{\hfill\blacklozenge}}

\DeclareMathOperator{\sign}{sign}

\newlength{\algitab}
\newcommand{\eot}{\hfill\ensuremath{\blacksquare}}
\newcommand{\bbm}{\begin{bmatrix}}
\newcommand{\ebm}{\end{bmatrix}}

\newcommand{\hlight}[2]{\textcolor{black}{#2}}

\begin{document}
\title{Learning Dictionaries from Physical-Based Interpolation for Water Network Leak Localization}

\author{Paul~Irofti,~\IEEEmembership{Member,~IEEE,}
        Luis~Romero-Ben,
        Florin~Stoican,~\IEEEmembership{Member,~IEEE,}
        and~Vicenç~Puig
\thanks{Paul Irofti is with the 
Research Center for Logic, Optimization and Security (LOS), Department of Computer Science, Faculty of Mathematics and Computer Science, University of Bucharest, 010014 Romania, 
and with the
Institute for Logic and Data Science, Popa Tatu 18, Bucharest, 010805 Romania,
e-mail: {paul@irofti.net}}%
\thanks{Luis Romero-Ben and Vicenç Puig are with the Institut de Robòtica i Informàtica Industrial, CSIC-UPC, Llorens i Artigas 4-6, 08028, Barcelona, Spain, e-mail: \{luis.romero.ben,vicenc.puig\}@upc.edu}
\thanks{Vicenç Puig is also with the Supervision, Safety and Automatic Control Research Center (CS2AC) of the Universitat Politècnica de Catalunya, Campus de Terrassa, Gaia Building, Rambla Sant Nebridi 22, Terrassa, 08222, Barcelona, Spain}
\thanks{Florin Stoican is with the  Deptartment of Automation Control and Systems Engineering, Politehnica University of Bucharest, 060042 Romania, e-mail: {florin.stoican@upb.ro}.}
}

\markboth{Journal of \LaTeX\ Class Files,~Vol.~14, No.~8, August~2021}%
{Shell \MakeLowercase{\textit{et al.}}: A Sample Article Using IEEEtran.cls for IEEE Journals}

\IEEEpubid{0000--0000/00\$00.00~\copyright~2021 IEEE}

\maketitle
\begin{abstract}
This article presents a leak localization methodology based on state estimation and learning. The first is handled by an interpolation scheme, whereas dictionary learning is considered for the second stage. The novel proposed interpolation technique exploits the physics of the interconnections between hydraulic heads of neighboring nodes in water distribution networks. Additionally, residuals are directly interpolated instead of hydraulic head values.  The results of applying the proposed method to a well-known case study (Modena) demonstrated the improvements of the new interpolation method with respect to a state-of-the-art approach, both in terms of interpolation error (considering state and residual estimation) and posterior localization.
\end{abstract}


\section{Introduction}\label{section:introduction}

\IEEEPARstart{T}{he} efficient management of water resources plays a crucial role in modern society, as clean water is vital for human life. Despite this, around four billion people suffer from severe water scarcity during at least one month per year \cite{Mekonnen2016}. This will only get worse in the near future, considering the continuous growth in the world-wide population, the climate change and the massification of urban areas \cite{Cohen2004}. 


Consequently, water utilities are greatly interested in developing methodologies for the efficient distribution of water. One of their main interests is the monitoring of leaks over water distribution networks (WDNs), which are estimated to exceed 126 billion cubic meters per year worldwide \cite{Liemberger2019}. This deeply impacts society, not only due to the economical and operational cost, but because of the environmental \cite{Xu2014} and sanitary \cite{Lechevallier2003} repercussions.  The research on this field has evolved from initial works \cite{Sterling1984,Pudar1992}, rapidly growing and leading to a wide range of methods (see \cite{RomeroBen2023} for a recent review about leak detection and localization methods).


Model-based schemes use calibrated hydraulic models (in terms of network properties and nodal consumption) to compare simulated hydraulic data with measurements from the real network. This approach goes back to the first works on the subject, when analysis of pressure sensitivity to the existence of leaks was studied \cite{Pudar1992}. The concept of sensitivity was further exploited in well-known articles, by means of a fault-sensitivity-matrix (FSM) that stores the effect of every possible leak on each junction \cite{Perez2014,Steffelbauer2022}. Additional techniques like the solution of the inverse problem, which consists of calibrating network parameters and demands from pressure/flow data, have also been widely used to detect/locate leaks 
\cite{Wu2010,Quinones2021}. Bayesian-based approaches \cite{Costanzo2014} or fuzzy logic \cite{Sanz2012} have also been explored. Their effectiveness is limited due to the diversity and complexity of existent networks \cite{Kim2016}, the difficult selection/calibration of the required hydraulic models \cite{Menapace2018} and the existence of modeling errors \cite{Blesa2018}. 


Recently, data-driven methods have started to be considered in order to circumvent the necessity of a hydraulic model, with the aim of avoiding the model-related drawbacks. Most of them rely on the exploitation of hydraulic data from deployed sensors and the network's topological features. For instance, interpolation-based techniques \cite{Soldevila2020,RomeroBen2022} estimate the values of unknown hydraulic variables from the set of measured ones; and graph-based methods \cite{Rajeswaran2017,Alves2022} exploit graph-related techniques (e.g. clustering) or properties to help solving the localization task. Their performance is satisfactory, particularly when the goal is to identify a leak search area over the network, even more when considering their independence of hydraulic models. Nevertheless, their accuracy tends to be lower in comparison to model-based schemes when their aim is to find the exact leak location. 
Moreover, the critical dependency on hydraulic measurements may lead to the necessity of a large set of sensors, increasing the operational cost.

The existing drawbacks of model-based and data-driven approaches, together with the development of data analysis and machine learning algorithms, led to the appearance of
mixed model-based/data-driven techniques. They use those algorithms to reduce the dependence on a hydraulic model, while maintaining the node-level accuracy.
Different algorithms, such as artificial neural networks (ANNs) \cite{Capelo2021}, support vector machine (SVM) \cite{AresMillan2021}, deep learning \cite{Romero2022} and dictionary learning (DL) \cite{IS17_ifac}; have been studied. 

\IEEEpubidadjcol

This paper develops a mixed model-based/data-driven methodology based on a novel hydraulic state interpolation technique, related to the one presented in \cite{RomeroBen2022}, and dictionary learning \cite{ISP20_jpc}.  This work continues the research presented in \cite{Irofti2022} \hlight{C1 - R4}{introducing a new physical-based interpolation method} \hlight{C6 - R11}{ and significantly improving on the simulation and related analysis}. A set of physical-based weights are derived from the WDN properties, namely pipe connectivity, lengths, diameters and roughness. These weights are used to pose a quadratic programming problem, which retrieves the pressure residuals (difference of pressure between leak and leak-free scenarios) for the complete network by means of the known values from sensorized nodes. Then, a group of these residuals are selected to be the samples provided to the learning stage, considering both the actual deployed sensors and a list of virtual sensors (VSs). The specific leak event linked to each selected residual is utilized to generate labels for the supervised training. 
Thus, the goal of combining interpolation and learning is the increase on the information provided to the DL approach (in comparison to merely supplying measurements), and the achievement of a satisfactory node-level localization accuracy that might be difficult for purely data-driven techniques while minimizing the usage of the hydraulic model (it is only necessary to obtain a set of leak samples). In this way, the strengths of each interpolation and learning/classification are boosted, while their weaknesses are diminished.

\noindent\textbf{Contributions:}

\textcolor{black}{The novel methodology presented in \cite{Irofti2022} has been updated and improved in this work, leading to a set of advantages related to performance and implementation:} 

\begin{itemize}
    \item \hlight{C1 - R4}{The methodology presented in \cite{Irofti2022} combined an existing interpolation method \cite{RomeroBen2022} with dictionary learning \cite{ISP20_jpc}. This interpolation scheme is} \hlight{C6 - R11}{based on a graph weighting approach that only considers structural features (pipe lengths). This may be desirable to reduce the scheme's information needs, but degrades the resulting interpolation performance. In this paper, a new physical-based weight generation process is introduced, whereas the quadratic programming problem posed in previous articles is reformulated to fit the derivations of the new physical-based formulation. In this way, both interpolation and localization errors are reduced in comparison to the results provided by the original interpolation.}
    \item\hlight{C3 - R4}{
    Preliminary results in \cite{Irofti2022} encouraged the use of VSs data} \textcolor{black}{
    obtained via interpolation in the dictionary learning process. 
    The main issue we faced revolved around
    the errors introduced by interpolated data, resulting in dictionaries modeling water networks
    that are different from the actual physical network under analysis
    (i.e. the process was not robust to interpolation noise). Tackling this artifact involved choosing only a few virtual sensors through human in the loop or grid-search techniques, because increasing the number of VSs would cause a loss of FDI performance.
    In this paper, the introduced theoretical approach producing analytical weights (used in the physical-based interpolation)
    allows any or all the non-sensorized nodes to be used as VSs in the dictionary learning process. Thus, the VSs are selected using purely data-driven sensor placement, unlike in \cite{Irofti2022}. The method presented in \cite{RomeroBen2022b} is adapted to seek the set of VSs that, added to the real sensors, minimizes a pre-defined distance-based metric, computed considering the distance from all the nodes to the sensors.
    }
    \item \textcolor{black}{Furthermore, the interpolation-learning scheme proves its capability of balancing out the main downside of a state-of-the-art data-driven methods (including the approach in \cite{Irofti2022}), namely the unreliable localization at node-level.}
\end{itemize}

\noindent\textbf{Notations.}
Vectors and matrices will be denoted in bold with lowercase and, respectively, uppercase letters.
Scalars will be denoted with plain lowercase letters.
For matrix $\bm{M}$ we denote $\bm{m}_i$ as its $i^{\text{th}}$ column and $m_{ij}$ as the entry of coordinates $i-j$. 
We denote approximations with a hat: $\bm{\hat{M}}$ is the approximation of $\bm{M}$. The notation $|\cdot|$ indicates cardinality when applied to sets, and absolute value otherwise. 

\noindent\textbf{Outline.}
The remainder of the paper is structured as follows.
Section \ref{section:preliminaries} presents 
the preliminary background regarding the considered interpolation techniques.
Section~\ref{section:approach} contains
the theoretical contribution consisting of 
the derivation of physically-based weights, and
their integration in the interpolation process
together with the leak isolation mechanism via dictionary learning.
The methodology from Section~\ref{section:methodology}
presents the integration of the interpolation and isolation process
together with the most important implementation aspects.
\textcolor{black}{Then, Section \ref{section:results} introduces the considered case study, presenting the specific implementation details of the method for the designed experiments. Obtained results for these experiments are subsequently exposed for both the interpolation and the localization performance.} Finally, Section \ref{section:conclusions} draws several conclusions about the proposed methodology, regarding its contributions, its applicability and future work-lines.


\section{Preliminaries}\label{section:preliminaries}



Arguably, all leak management methodologies crucially depend on the accuracy and quantity of hydraulic information, which can be obtained from hydraulic models and/or a sufficiently large and well-distributed set of sensors. However, these can be limited by the operational and economical costs of creating/calibrating models and installing/maintaining sensors. \hlight{C3 - R11}{A compromise solution lies in interpolation algorithms, which estimate the WDN hydraulic state exploiting the available WDN structure and hydraulic data collected from a reduced set of sensors\cite{Soldevila2021,RomeroBen2022}. Let us underscore key aspects regarding these two sources of information of the network:}

\begin{itemize}
\item \textcolor{black}{The WDN structure can be modeled through a graph $\mathcal{G}=(\mathcal{V},\mathcal{E})$, where $\mathcal{V}$ denotes the set of nodes (reservoirs and junctions); and $\mathcal{E}$ is the set of edges (pipes). An arbitrary node is denoted as $\mathscr{v}_i\in\mathcal{V}$, while an arbitrary edge is referred to as $\mathscr{e}_k = (\mathscr{v}_i,\mathscr{v}_j)\in\mathcal{E}$, representing the link between nodes $\mathscr{v}_i$ and $\mathscr{v}_j$, considering the first as the source and the second as the sink. The network is composed by $n=\lvert \mathcal{V}\rvert$ nodes and $m = \lvert \mathcal{E}\rvert$ edges. In practice, network connectivity and pipe lengths are normally available to water utilities (considering tools like geographical information systems (GIS) \cite{Taher1996}), while pipe diameters and roughness can be approximately estimated.}

\item \textcolor{black}{Normally, interpolation-based schemes consider nodal pressure or, commonly, the hydraulic head (nodal pressure plus elevation), as representative of the network state. This is justified by the fact that leaks manifest as pressure drops, and pressure sensors are cheaper and easier to install than other metering devices, such as flow sensors. Thus, a set of $n_{\zeta}$ pressure sensors must be installed across the network, leading to a set of measured hydraulic heads. Note that the heads of connected nodes are related by the flow conveyed from source to sink through a non-linear equation such as the Hazen-Williams formula \cite{sanz2016demand}. This relation implies that the hydraulic head of the source will always be higher than the one of the sink, if no active elements such as pumps are installed.}
\end{itemize}

All these aspects were considered to develop the Graph-based State Interpolation (GSI) technique from \cite{RomeroBen2022}. The core idea was to approximate the aforementioned non-linear node head dependency by a linear relation, 
\begin{equation}\label{eq:1}
    \hat{\psi}_i = \frac{1}{\phi_i}\bm{\omega}_i\bm{\hat{\psi}},
\end{equation}
\noindent where $\bm{\hat{\psi}}\in \mathbb{R}^{n}$ stands for the complete state vector, which approximates the actual hydraulic head vector through a mix of estimated and known values. 
\hlight{C1 - R11}{$\bm{\omega_i}$ is the \textit{i-th} row of the weighted adjacency matrix $\bm{\Omega}\in \mathbb{R}^{n \times n}$}, 
which encodes the weights of the connection between nodes. Specifically, 
GSI (as defined in \cite{RomeroBen2022}) 
gives higher weights to closer neighbors 
by taking them as the inverse of the actual pipe lengths     
\begin{equation}\label{eq:weighting_selection}
    \begin{gathered}
        \omega_{ij} = \begin{cases} \frac{1}{p_{k}}, & \mbox{if } \mathscr{e}_k = (\mathscr{v}_i,\mathscr{v}_j) \in \mathcal{E}, \\
        0, & \mbox{otherwise, } 
        \end{cases}       
    \end{gathered}
\end{equation}
    
\noindent where $\bm p\in \mathbb{R}^{m}$ is the pipe lenghts vector.
Moreover, $\phi_i$ stands for the \textit{i-th} element of the diagonal of the degree matrix $\bm{\Phi}\in \mathbb{R}^{n \times n}$, directly obtained from $\bm{\Omega}$ by taking $\phi_i=\sum_{j=1}^{n} \omega_{ij}$.

The flow directionality in WDNs is imposed by the non-linear relation between hydraulic heads of adjacent nodes, with flow traversing from source to sink. 
\hlight{C5 - R11}{This directionality can be encoded by means of the edge-node incidence matrix $\bm \Lambda\in \mathbb R^{m\times n}$, whose entries are defined as follows}

\begin{equation} \label{eq:incidence}
\begin{gathered}
\lambda_{kj} = \begin{cases} \hphantom{-}1, & \mbox{if } \mathscr{e}_k = (\mathscr{v}_i,\mathscr{v}_j)\in \mathcal{E},\\
-1, & \mbox{if } \mathscr{e}_k = (\mathscr{v}_j,\mathscr{v}_i)\in \mathcal{E},\\
\hphantom{-}0, & \mbox{otherwise.} 
\end{cases} 
\end{gathered}
\end{equation}

Nevertheless, 
water utilities do not usually have access to a sufficient amount of flow data 
to obtain an accurate hydraulic-based incidence matrix. Thus, GSI considers a structural approach to construct an approximated incidence matrix $\bm{\hat{\Lambda}}$, as detailed in Algorithm \ref{alg:b-generation}. 
\begin{algorithm}[!ht]
\caption{Approximated incidence matrix generation}
\label{alg:b-generation}
\begin{algorithmic}[1]
\REQUIRE {$\mathcal{G}=(\mathcal{V},\mathcal{E}), \mathcal{V}_{\mathcal{R}} \subseteq \mathcal{V}$}
\STATE Initialize $\bm U = \bm 0_{(n\times n)}, \bm{\hat{\Lambda}} = \bm 0_{(m\times n)}$
\STATE Compute $\mathcal{V}_{\mathcal{N}} = \mathcal{V} \setminus \mathcal{V}_{\mathcal{R}}$
\FORALL{$\mathscr{v}_{\mathcal{R}} \in \mathcal{V}_{\mathcal{R}}$} 
\FORALL{$\mathscr{v}_{\mathcal{N}} \in \mathcal{V}_{\mathcal{N}}$} 
\STATE $\mathfrak{\bm s}^{path} = shortest\_path(\mathcal{G},\mathscr{v}_{\mathcal{R}},\mathscr{v}_{\mathcal{N}})$
\FOR{$i = 1,...,|\mathfrak{\bm s}^{path}|-1$} 
\STATE $\bm U(\mathfrak{s}_i^{path},\mathfrak{s}_{i+1}^{path}) = \bm U(\mathfrak{s}_i^{path},\mathfrak{s}_{i+1}^{path}) + 1$
\ENDFOR
\ENDFOR
\ENDFOR
\FORALL{$\mathscr{e}_k = (\mathscr{v}_i,\mathscr{v}_j) \in \mathcal{E}$}
\IF{$\bm U(\mathscr{v}_i,\mathscr{v}_j) > \bm U(\mathscr{v}_j,\mathscr{v}_i)$} 
\STATE $\hat{\lambda}_{ki} = -1; \hat{\lambda}_{kj} = 1$ 
\ELSE 
\STATE $\hat{\lambda}_{ki} = 1; \hat{\lambda}_{kj} = -1$
\ENDIF
\ENDFOR
\RETURN $\bm{\hat{\Lambda}}$
\end{algorithmic}
\end{algorithm}

Note that $\mathcal{V}_{\mathcal{R}}, \mathcal{V}_{\mathcal{N}}$ stand for the sets of reservoirs and inner nodes respectively\footnote{The algorithm exploits the fact that water must flow from the reservoirs to the rest of the water-demanding nodes.}, and $shortest\_path(\cdot,\cdot,\cdot)$ generates the shortest path in a graph between two nodes \cite{Chen2003}, encoded as the output sequence/path $\mathfrak{\bm s}^{path}$ \hlight{C4 - R11}{(step 5)}. Additionally, $\bm U \in \R^{n\times n}$ stores the number of times that 
a pipe has been considered in 
each $\mathfrak{\bm s}^{path}$ \hlight{C4 - R11}{(step 7)}. Thus, the most used direction for each pipe is considered and encoded into the generated incidence matrix $\bm{\hat{\Lambda}}$ \hlight{C4 - R11}{(steps 13 and 15). In summary, the heuristic under consideration poses that if a greater number of paths have traversed a pipe in a specific direction, it is reasonable to assume that the flow within the pipe will predominantly follow that same direction.}

Both weight selection and directionality requirements are integrated in an optimization problem whose solution is the interpolated graph state (approximated hydraulic heads), \cite{RomeroBen2022}:
\begin{subequations}
\label{eq:2}
\begin{align}
\label{eq:2_a}
\min_{\bm{\hat{\psi}}} \quad & \frac{1}{2}\big[\bm{\hat{\psi}}^T\bm{L}\bm{\Phi}^{-2}\bm{L}\bm{\hat{\psi}}+\mu \gamma ^2\big],\\
\label{eq:2_b}\textrm{s.t.} \quad & \bm{\hat{\Lambda}} \bm{\hat{\psi}}\leq \bm{1}_{n}\cdot\gamma,\\
\label{eq:2_c}& \gamma > 0,\\
\label{eq:2_d}&\bm{Z}\bm{\hat{\psi}}=\bm{\hat{\psi}}_{\zeta}.  
\end{align}
\end{subequations}



\hlight{C3 - R11}{The first term of objective \eqref{eq:2_a} comes from enumerating a relaxation of \eqref{eq:1} for all nodes in $\mathcal V$. That is, we take the quadratic error of the left and right-sides of \eqref{eq:1}:}

\begin{equation}\label{eq:minim-obj-func}
    \sum_{i=1}^{n}\Big[\hat{\psi}_i - \frac{1}{\phi_i}\bm{\omega}_{i}\bm{\hat{\psi}}\Big]^2 = \bm{\hat{\psi}}^T\bm L\bm\Phi^{-2}\bm L\bm{\hat{\psi}},
\end{equation}



\noindent \textcolor{black}{which becomes the term in \eqref{eq:2_a} when taking $\bm L=\bm\Phi - \bm\Omega$, i.e., the Laplacian of $\mathcal G$; and after standard linear algebra manipulations.} 

\textcolor{black}{The second term of (\ref{eq:2_a}) and the constraint in (\ref{eq:2_b}) are directly related to the flow directionality.  In vector $ {\bm\Lambda} \bm\psi$, the k-th term contains ${\lambda}_{ki}\psi_i+{\lambda}_{kj}\psi_j$, which is negative as per the construction in \eqref{eq:incidence}, and hence we expect ${\bm\Lambda} \bm\psi\leq 0$. However, we do not have access to the actual flow directionalities (we have instead $\hat{\bm \Lambda}\hat{\bm \psi}$), so we relax the initial inequality to the form from \eqref{eq:2_b} via the positive slack term $\gamma$ (\ref{eq:2_c}), which is then minimized through \eqref{eq:2_a}. The relative importance of interpolation error and flow directionality consistency is tunned by parameter $\mu$. }


\textcolor{black}{Finally, the available hydraulic information $\bm{\hat{\psi}}_{\zeta}$ is 
provided through the equality constraint in (\ref{eq:2_d}), which restricts the state of the sensorized nodes to be equal to the actual hydraulic heads at the corresponding junctions. Matrix $\bm Z \in \mathbb{R}^{n_{\zeta}\times n}$ is defined so that $z_{oj} = 1$ only if the $o-th$ sensor is located in node $\mathscr{v}_j$, and 0 otherwise.}  

\begin{remark}\label{remark:sGSI}
 The interpolation of vertex values over a graph following smoothing strategies was studied in the past \cite{Zhu2003}. These techniques solve a problem that intends to minimize an analogue expression to $\frac{1}{2}\big[\bm{\hat{\psi}}^T\bm{L}\bm{\hat{\psi}}\big]$. In this manner, the harmonic property of functions in graphs is explicitly pursued, i.e., the states of adjacent nodes are aimed to be as similar as possible (GSI does not impose this behavior, giving room to a higher degree of freedom to fulfill the added constraints). 
 \eor
\end{remark}

\begin{remark}
Since GSI solves the quadratic problem \eqref{eq:2}, it has polynomial complexity regarding problem dimension (e.g., solving $N$ problems with the interior point method gives $O(n^{3.5}N)$).\eor
\end{remark}



\section{Proposed Approach}\label{section:approach}

\subsection{Physical-based interpolation}

The weighting selection in GSI, described by \eqref{eq:weighting_selection}, arises from the aim of minimizing the information requirements regarding the network characteristics. To this end, the network connectivity and pipe lengths are its unique inputs related to the WDN topology. However, this degrades the interpolation performance due to the approximations introduced by the weighted linear expression that substitutes the non-linear equation relating the hydraulic heads of adjacent nodes.

Thus, we aim to derive a new weighting approach to improve the interpolation, which better approximates the non-linear nature of the steady-state head-flow equations. 
To this end, let us recall the empirical Hazen-Williams formula\footnote{We redefine the pipe-related variables as follows: considering that $\mathscr{e}_k=(\mathscr{v}_i,\mathscr{v}_j)$, the $k-th$ element of a pipe-related variable $\bm x \in \mathbb{R}^m$ would be denoted as $x_{(i,j)}$ for $i,j=1,2,..,n$ instead of $x_k$ for $ k=1,2,...,m$ (thus, each pair $(i,j)$ is directly related to an index $k$). This improves the readability of the section by expressing all the formulas in terms of node indices.  
}

\be
\label{eq:hw}
\psi_i-\psi_j = \frac{10.67\cdot p_{(i,j)}}{r_{(i,j)}^{1.852}\cdot \delta_{(i,j)}^{4.87}}\cdot f_{(i,j)}\cdot | f_{(i,j)}|^{0.852},
\ee
where $p_{(i,j)} = p_k$ is the length in $[m]$, $\delta_{(i,j)}=\delta_k$ is the diameter in $[m]$, $r_{(i,j)}=r_k$ is the adimensional pipe roughness coefficient and the flow $f_{(i,j)} = f_k, \forall \mathscr{e}_k \in \mathcal{E}$ is measured in $[m^3/s]$. Noting that $\sign( f_{(i,j)})=\sign( \psi_i- \psi_j)$ and that $x=|x|\cdot \sign(x)$, we reformulate \eqref{eq:hw} into
\begin{equation}
    \label{eq:hw_power}
    f_{(i,j)}= \sigma_{(i,j)}^{0.54}\sign(\psi_i-\psi_j)| \psi_i-\psi_j|^{0.54},
\end{equation}
where $\mathbf \sigma_{(i,j)}=\sigma_k=(r_{(i,j)}^{1.852}\cdot \delta_{(i,j)}^{4.87})/(10.67\cdot p_{(i,j)})$ denoting the pipe conductivity and considering that $0.54\approx 1/1.852$.

In order to facilitate the manipulation of $\sign(\psi_i-\psi_j)$ through the rest of the paper, we convert the previously presented edge-node incidence matrix $\bm{\Lambda}$
into node-node incidence matrix $\bm{B}\in\mathbb{R}^{n\times n}$,
which is defined as
\be
\label{eq:b-adj}
 b_{ij}=\begin{cases}\hphantom{-}1,& \psi_i\geq  \psi_j;\\ \hphantom{-}0,& \textrm{if nodes $i$ and $j$ are not connected;}\\ -1,& \psi_i<  \psi_j.\end{cases}
\ee

\begin{remark}
To clarify the relationship between $\bm{\Lambda}$ and $\bm{B}$,
define $\bm{\Lambda}_{\text{out}} = \frac{\bm{\Lambda} + \abs{\bm{\Lambda}}}{2}$
and
$\bm{\Lambda}_\text{in} = -\frac{\bm{\Lambda} - \abs{\bm{\Lambda}}}{2}$
as representing the inner and outer edges
such that
$\bm{\Lambda} = \bm{\Lambda}_{\text{out}} - \bm{\Lambda}_{\text{in}}$.
Let $\bm{A} = \bm{\Lambda}_{\text{out}}^T\bm{\Lambda}_{\text{in}}$
which gives $\bm{B} = \bm{A}^T - \bm{A}$.
With a few linear algebra calculations
we arrive at the straight-forward relation
$\bm{B} =
\frac12 \left(\abs{\bm{\Lambda}}^T \bm{\Lambda} - \bm{\Lambda}^T \abs{\bm{\Lambda}} \right)$.\eor
\end{remark}


Considering that the balance of inflows and outflows passing throughout a node 
of index $i$  can be expressed as:

\begin{equation}
    \label{eq:b}
    \sum\limits_{j:\: b_{ij}\neq 0} f_{(i,j)}=c_i,
\end{equation}
where 
$c_i$ is the consumption (the normal user demand and, potentially, leaks), by introducing (\ref{eq:hw_power}) in (\ref{eq:b}), we arrive at
\begin{equation}
\sum\limits_{j:\: b_{ij}>0}\mkern-16mu \sigma_{(i,j)}^{0.54}(\psi_i-\psi_j)^{0.54}\mkern-4mu-\mkern-16mu\sum\limits_{j:\: b_{ij}<0}\mkern-16mu \mathbf \sigma_{(i,j)}^{0.54}(\psi_j-\psi_i)^{0.54} = c_i,
\end{equation}
or, equivalently stated,
\begin{equation}
\label{eq:wdn_hw}
\sum\limits_{j:\: b_{ij}\neq 0} b_{ij} \sigma_{(i,j)}^{0.54}\left[b_{ij}(\psi_i-\psi_j)\right]^{0.54} = c_i.
\end{equation}

Recall that head values in \eqref{eq:wdn_hw} corresponding to reservoirs are fixed. Without loss of generality, these known head values may be introduced into equations \eqref{eq:wdn_hw} such that only the `true' variables, the junction node heads, remain. In any case, solving the group of equations \eqref{eq:wdn_hw} to retrieve the steady-state node heads is not an easy task. This is why hydraulic simulators like EPANET  \cite{Epanet:2000} are used, as they provide numerical solutions instead of attempting to derive analytic ones.

Thus, our goal is not to solve the WDN problem but rather to exploit the relations underlined by \eqref{eq:wdn_hw} to find first order (linear) dependencies between the current junction node's head, $\psi_i$, and its neighbors, $\psi_j$. To do so, we apply the \emph{implicit function theorem} \cite{narasimhan1985analysis}, as detailed next.

\begin{lemma}
\label{lem:implicit}
For the i-th node relation \eqref{eq:wdn_hw}, assuming that there exist $\bar{\psi}_i$, $\bar{\psi}_j,\forall j\in \J = \{j\; |\; b_{ij}\neq 0\}$ which verify it, there:
\begin{enumerate}[i)]
    \item exists a local (in the neighborhood of $\bar{\psi}_i$, $\bar{\psi}_j$) and explicit dependence between its head ($\psi_i$) and its neighbors' heads ($\psi_j,\forall j \in \J$),
    \begin{equation}
    \label{eq:wdn_g}
        \psi_i=\mathbf g_i(\bm{\psi}_{\J}), 
    \end{equation}
     with 
     $\bm{\psi}_{\J}=\bbm \dots & \psi_j & \dots\ebm^\top\in \mathbb R^{|\J|}$ grouping all neighboring node's head values;
     \item and function's $\mathbf g_i:\mathbb R^{|\J|}\mapsto \mathbb R$ derivatives are given by\footnote{Note that $k$ is used here to index the neighbors of the $i-th$ node during the summation, and it is not the same as the $k$ used before to denote the k-th edge. This must be considered for subsequent analogue usages of $k$ as an index.}
     \begin{equation}
     \label{eq:wdn_gderiv}
         \frac{\partial \mathbf g_i(\bm{\psi}_{\J})}{\partial\psi_j}=\frac{\sigma_{(i,j)}^{0.54}\left[b_{ij}(\psi_i-\psi_j)\right]^{-0.46}}{\sum\limits_{k:\: b_{ik}\neq 0} \mkern-8mu \sigma_{ik}^{0.54}\left[b_{ik}(\psi_i-\psi_k)\right]^{-0.46}},
     \end{equation}
\end{enumerate}
for all $j\in \J$. \eot
\end{lemma}
\begin{proof} See the appendix.
\end{proof}


\begin{prop}
\label{prop:weights}
For the current node $i$, consider the head values $\bar{\psi}_i$, $\bar{\psi}_j,\forall j\in \J$ which verify \eqref{eq:wdn_hw}. Then, $\mathbf g_i(\cdot)$, defined as in \eqref{eq:wdn_g} is approximated by the linear relation 
\begin{equation}
\label{eq:barhi_firstorder}
    \psi_i=\mathbf g_i(\bm{\psi}_{\J})\approx \bar{\psi}_i+\sum\limits_{j:\: b_{ij}\neq 0}\eta_{ij}\cdot (\psi_j-\bar{\psi}_j),
\end{equation}
where the weights $\eta_{ij}$ are given by
\begin{equation}
\label{eq:wdn_omega}
    \eta_{ij}(\bar{\psi}_i,\bar{\psi}_j)=\frac{\sigma_{(i,j)}^{0.54}\left[ b_{ij}(\bar{\psi}_i-\bar{\psi}_j)\right]^{-0.46}}{\sum\limits_{k:\:  b_{ik}\neq 0} \sigma_{(i,j)}^{0.54}\left[ b_{ik}(\bar{\psi}_i-\bar{\psi}_k)\right]^{-0.46}},
\end{equation}
for all $j\in \J$. \eot
\end{prop}
\begin{proof}
As per the first part of Lemma~\ref{lem:implicit}, we have that there exists $\mathbf g_i$ defined as in \eqref{eq:wdn_g}. We approximate this function by the first two terms of its Taylor expansion centered around $\bar{\psi}_i$, $\bar{\psi}_j,\forall j\in \J$:
\begin{equation}
\label{eq:wdn_gtaylor}
    \mathbf g_i(\bm{\psi}_{\J})\approx\bar{\psi}_i+\mkern-8mu\sum\limits_{j:\: b_{ij}\neq 0}\left[\left.\frac{\partial \mathbf g_i(\bm{\psi}_{\J})}{\partial\psi_j}\right|_{(\bar{\psi}_i,\bar{\psi}_j)}\cdot (\psi_j-\bar{\psi}_j)\right].
\end{equation}
Introducing relations \eqref{eq:wdn_gderiv} into \eqref{eq:wdn_gtaylor} leads to:
\begin{equation}
\label{eq:wdn_gtaylor2}
    \mathbf g_i(\bm{\psi}_{\J})\approx\bar{\psi}_i+\frac{\sum\limits_{j:\:  b_{ij}\neq 0}\sigma_{(i,j)}^{0.54}\left[ b_{ij}(\bar{\psi}_i-\bar{\psi}_j)\right]^{-0.46}\cdot (\psi_j-\bar{\psi}_j)}{\sum\limits_{k:\:  b_{ik}\neq 0} \mkern-8mu \sigma_{(i,j)}^{0.54}\left[ b_{ik}(\bar{\psi}_i-\bar{\psi}_k)\right]^{-0.46}}.
\end{equation}
Rearranging the terms of \eqref{eq:wdn_gtaylor2} directly leads to the values from \eqref{eq:wdn_omega}, thus concluding the proof.
\end{proof}

\subsection{Graph-based state interpolation reformulation}


The weights $\eta_{ij}$ in \eqref{eq:wdn_omega} can be structured like those in \eqref{eq:1},

\begin{equation}
\label{eq:c1}
    \eta_{ij} = \frac{1}{\phi_{i}}\omega_{ij},
\end{equation}

\noindent where

\begin{equation}
\label{eq:c2}
\begin{gathered}
    \omega_{ij}(\bar{\psi}_i,\bar{\psi}_j) = \sigma_{(i,j)}^{0.54}\left[ b_{ij}(\bar{\psi}_i-\bar{\psi}_j)\right]^{-0.46}, \\
    \phi_{i}(\bar{\psi}_i) = \sum\limits_{k:\:  b_{ik}\neq 0} \sigma_{(i,j)}^{0.54}\left[ b_{ik}(\bar{\psi}_i-\bar{\psi}_k)\right]^{-0.46}.
    \end{gathered}
\end{equation}

Thus, $\phi_{i} = \sum\limits_{k:\:  b_{ik}\neq 0} \omega_{ik}$ holds, letting us consider $\omega_{ij}$ as the $i-j$ entry of a weighted adjacency matrix $\mathbf \Omega$ that encodes the relation between nodes $\mathscr{v}_i$ and $\mathscr{v}_j$, with $\phi_{i}$ as the degree of node $\mathscr{v}_i$. Substituting (\ref{eq:c1}) into (\ref{eq:barhi_firstorder}), we obtain

\begin{equation}
\label{eq:5}
    \psi_i=\mathbf g_i(\bm{\psi}_{\J})\approx \bar{\psi}_i+ \frac{1}{\phi_{i}}
 \sum\limits_{j:\: b_{ij}\neq 0}\omega_{ij}\cdot (\psi_j-\bar{\psi}_j),
\end{equation}

\noindent which is written in matrix form as

\begin{equation}
\label{eq:6}
    \bm \psi - \bar{\bm \psi} \approx \mathbf \Phi^{-1} \mathbf \Omega(\bm \psi - \bar{\bm\psi}).
\end{equation}


To arrive at an interpolation procedure similar with \eqref{eq:2} we replace \eqref{eq:6} with 

\begin{equation}
\label{eq:7}
    \hat{\bm \psi} - \hat{\bar{\bm \psi}} = \mathbf{\Delta\hat{\bm \psi}}= \mathbf \Phi^{-1} \mathbf \Omega(\hat{\bm \psi} - \hat{\bar{\bm\psi}}).
\end{equation}

\noindent which uses the estimated hydraulic head vector instead.

The leak-free estimated hydraulic heads are selected to perform the role of $\hat{\bar{\bm\psi}}$, so that $\mathbf{\Delta\hat{\bm \psi}}$ denotes the estimated pressure residuals. Thus, the minimization of the difference between the two sides of (\ref{eq:7}), which is analogue to (\ref{eq:minim-obj-func}), 
implies that the formulation proposed in (\ref{eq:2}) must be modified to consider the change in the interpolated variable, i.e., the method interpolates pressure residuals instead of heads

\begin{subequations}
\label{eq:8}
\begin{align}
\label{eq:8_a}
\min_{\bm{\Delta\hat{\bm \psi}}} \quad & \frac{1}{2}\big[\bm{\Delta\hat{\bm \psi}}^T\bm{L}\bm{\Phi}^{-2}\bm{L}\bm{\Delta\hat{\bm \psi}}\big],
\\
\label{eq:8_b}\textrm{s.t.} \quad &\bm{Z}\bm{\Delta\hat{\bm \psi}}=\bm{\Delta\hat{\bm \psi}}_{\zeta}. 
\end{align}
\end{subequations}

Note that the directionality term of the optimization problem \eqref{eq:2} no longer explicitly appears, as this direction cannot be imposed for the residuals. Nevertheless, it is indirectly included during the derivation of the weights (see (\ref{eq:c2})). The procedure of solving the problem posed in (\ref{eq:8}) is henceforward referred to as AW-GSI (analytical weights graph-based interpolation).

\subsection{Leak Localization via Dictionary Learning}\label{subsubsection:DL}

We achieve leak localization via dictionary learning classification on the available interpolated data.
The training phase of dictionary learning (DL) is performed on the samples generated from the application of AW-GSI to the available pressure residuals.
First, we briefly overview the DL algorithms and their relation to fault isolation.

Sparse representations~\cite{Elad10_book} are used to recover information from noisy samples through the use of a redundant basis,
often called dictionary.
For WDN, the samples are constituted by residuals from sensor readings and interpolated information, for various nominal and leaky scenarios taking place across the network nodes.
Given $n_{ts}$ total sensors (physical and virtual) and $n_{samp}$ scenarios,
we collect this data in the matrix $\bm{Y} \in \R^{n_{ts}\times n_{samp}}$
with which we aim to recover the node where the fault took place. In our case, $\bm{Y}$ would be filled with the entries $\bm{\Delta\hat{\bm \psi}}$, obtained from applying AW-GSI for all the considered leak scenarios and their corresponding data samples.

Given a dictionary $\bm{D}\in \R^{n_{ts}\times n_{atom}}$
consisting of $n_{atom}$ columns, also called atoms,
we obtain the sparse representations $\bm{X} \in \R^{n_{atom}\times n_{samp}}$.
Each column in $\bm{X}$ is $s$ sparse.
To maximize precision,
we specialize the dictionary for each WDN 
through the process of dictionary learning~\cite{DL_book}.
Here we are not solely interested in optimum representations,
but also in finding the dictionary best fitted for the WDN at hand,
\be
\begin{aligned}
& \underset{\bm{D}, \bm{X}}{\min}
& & \norm{\bm{Y}-\bm{DX}}_F^2, \\
& \text{s.t.}
& & \norm{\bm{x}_\ell}_{0} \leq s,\ \ell = 1:n_{samp}, \\
& & & \norm{\bm{d}_j} = 1, \ j = 1:n_{atom}.
\end{aligned}
\label{dict_learn}
\ee
where $\norm{\cdot}_0$ is the pseudo-norm counting the number of nonzeros. 
The goal is
to obtain a specialized dictionary $\bm{D}$
such that when we are given a new measurement $\bm{y}$
we can closely approximate it as $\bm{y}\approx\bm{Dx}$
through the use of a sparse representation algorithm.
In this paper we will be using 
Orthogonal Matching Pursuit (OMP)~\cite{PRK93omp}
for this task.

Model-free identification of a leaky node from the set of network nodes
is a classification task.
To improve DL classification,
we will be learning three separate dictionaries:
$\bm{D}$ to approximate the data,
$\bm{W}$ to classify the data,
and
$\bm{A}$ to specialize atom blocks per each class.
A class is a network node in our case.
This approach is called Label Consistent K-SVD (LC-KSVD)~\cite{JLD13}
and it extends \eqref{dict_learn} to simultaneously perform DL on all three dictionaries
\be 
\min_{\bm{D},\bm{W}, \bm{A}, \bm{X}} \|\bm{Y}-\bm{DX}\|_{F}^2 + \alpha \|\bm{H}-\bm{WX}\|_{F}^2
+ \beta \|\bm{Q}-\bm{AX}\|_{F}^2.
\label{opt_label_cons}
\ee
As it can be seen,
each dictionary is trained on different data.
$\bm{H} \in \R^{n_{class}\times n_{samp}}$ represents the label matrix,
where column $h_i$ has 1 set in the class corresponding to the leaky node in measurements $y_i$ and zeros elsewhere.
$\bm{Q} \in \R^{n_{atom}\times n_{samp}}$ enforces atom sets per class,
where column $q_i$ has ones in the positions corresponding to atoms
that should represent the class to which $y_i$ belongs
and zeros elsewhere.
Please note that in this process we start with
$\bm{Y}$ only containing sensor information,
and add the labels $\bm{H}$ to help identify non-sensorized nodes.

After the learning process, the dictionaries are used to perform fault detection and identification (FDI) on new sensor measurements $\bm{y}$ in a two-step process:
sparse representation $\bm{y}\approx\bm{Dx}$
and classification $i=\arg\max_j\left(\bm{Wx}\right)_j$,
thus obtaining the leaky node $i$.

\begin{remark}\label{remark:DL_complexity}
Looking at computational complexity,
the classification effort evolves around the sparse representation step
where OMP uses $O(sn_{ts}n_{atom})$ instructions for each data-item.
Dictionary learning consists of $K$ rounds of sparse representation 
and dictionary update.
LC-KSVD uses $O(sn_{ts}^2n_{samp})$ instructions to update the dictionary,
sparse representation OMP performance over one data-item uses $O(sn_{ts}n_{atom})$ operations, whereas the dictionary update in LC-KSVD has a complexity of $O(sn_{ts}^2n_{samp})$.\eor
\end{remark}

The choice in dictionary learning is motivated by our success in the past \cite{IS17_ifac,ISP20_jpc}, where DL was used without interpolation to tackle FDI. For large WDNs we extended this approach for online semi-supervised learning in~\cite{IB19_toddler}.

\section{Methodology overview}\label{section:methodology}

Once the core problem of AW-GSI has been defined, let us develop the operational flow of this leak localization methodology that combines interpolation and dictionary learning.

\subsubsection{Data generation}\label{subsubsection:data_generation}

The application of a supervised learning-based technique implies the necessity of sufficiently rich sets of samples and labels. In our case, the labels correspond to the different nodes of the network that may leak. Meanwhile, the samples are data vectors containing information about the state of the actual and virtual sensors at different scenarios and conditions. 
Both samples and labels are required from all possible leaks. To obtain this data, available historical leak datasets from water utilities can be useful. Besides, leak experiments can be executed over the field in order to gain data about the behaviour of the network on those abnormal events.\footnote{Note that nodes whose leaks cause similar effects over the network may be grouped into a common leak group, effectively reducing the number of labels in the learning phase, and hence the number of necessary experiments.} If these sources of information are not available, a hydraulic model of the network can be used to generate the necessary leak events. 

Several aspects/parameters need to be taken into account about the data generation process: 

\begin{itemize}
    \item The number of considered leaks, i.e., different classes/labels, $n_{class} \in \N$.
    \item The number of time instants per leak, $n_t \in \N$.
    \item The number of considered leak sizes, $n_{ls} \in \N$.
\end{itemize}

Then, the complete leak dataset is going to be composed by $n_{samp} = n_{class}n_tn_{ls}$ samples of length $n_{\zeta}$, i.e., $\hat{\bm\Psi}_{\zeta} \in \mathbb{R}^{n_{\zeta}\times n_{samp}}$. Additionally, note that the dataset must be divided into training and testing sets for the sake of the dictionary learning phase, i.e., $\bm\hat{\bm\Psi}^{train} \in \mathbb{R}^{n_{\zeta}\times n_{train}}$ and $\bm\hat{\bm\Psi}^{test} \in \mathbb{R}^{n_{\zeta}\times n_{test}}$, with $n_{samp} = n_{train} + n_{test}$. 

\begin{remark}
Apart from the leak dataset, a leak-free dataset is necessary because DL operates over residuals, which encode the difference in head state between nominal and leak conditions. Thus, each entry of the nominal set $\hat{\bar{\bm\Psi}}_{\zeta} \in \mathbb{R}^{n_{\zeta}\times n_{samp}}$ must be obtained with network boundary conditions (e.g. water inlet pressure, nodal demands, etc.) similar to its analogue leak entry in $\hat{\bm\Psi}_{\zeta}$, with the aim of minimizing possible state differences that are not related to the leak but to changes in consumer demands. 
The nominal dataset is also divided into training and testing, yielding $\hat{\bar{\bm\Psi}}^{train} \in \mathbb{R}^{n_{\zeta}\times n_{train}}$ and $\hat{\bar{\bm\Psi}}^{test} \in \mathbb{R}^{n_{\zeta}\times n_{test}}$.
\eor
\end{remark}

\subsubsection{AW-GSI }\label{subsubsection:aw-gsi}

The quadratic programming problem in (\ref{eq:8}) is the base of AW-GSI. Let us review the key points of the interpolation operation, summarized by Algorithm \ref{alg:aw-gsi}.

\begin{algorithm}[t]
\caption{AW-GSI}
\label{alg:aw-gsi}
\begin{algorithmic}[1]
\REQUIRE {$\hat{\bar{\bm\psi}}_{\zeta}, \hat{\bm\psi}_{\zeta}\in\mathbb{R}^{n_{\zeta}}, \bm \sigma \in \mathbb{R}^{m}$}
\STATE Compute $\hat{\bar{\bm\psi}}^{GSI}$ from $\hat{\bar{\bm\psi}}_{\zeta}$ solving (\ref{eq:2}) [Remark \ref{remark:sGSI}]
\STATE Compute $\bm \hat{\bm{B}}$ from $\hat{\bar{\bm\psi}}^{GSI}$ using (\ref{eq:b-adj})
\STATE Compute $\bm \Omega^{AW}, \bm \Phi^{AW}$ from $\hat{\bm{B}}, \bm \sigma, \hat{\bar{\bm\psi}}^{GSI}$ using (\ref{eq:c2})
\STATE Compute $\bm L^{AW} = \bm \Phi^{AW} - \bm \Omega^{AW}$
\STATE Compute $\bm\Delta\hat{\bm\psi}_{\zeta} = \hat{\bm\psi}_{\zeta} - \hat{\bar{\bm\psi}}_{\zeta}$
\STATE Compute $\bm\Delta\hat{\bm\psi}$ from $\bm\Delta\hat{\bm\psi}_{\zeta}, \bm L^{AW}, \bm \Phi^{AW}$ solving (\ref{eq:8})
\RETURN $\bm\Delta\hat{\bm\psi}$
\end{algorithmic}
\end{algorithm}

First, the selection of $\hat{\bar{\bm \psi}}$ should be defined. Considering that the learning phase is fed with pressure residuals, 
a smart choice would be to select nominal/leak-free values to be represented by $\hat{\bar{\bm \psi}}_{\zeta}$. The standard GSI problem in (\ref{eq:2}) is solved to obtain $\hat{\bar{\bm\psi}}^{GSI}$ from $\hat{\bar{\bm\psi}}_{\zeta}$, but considering the cost function stated in Remark \ref{remark:sGSI}, in order to explicitly pursue the state smoothness over the neighborhood. The complete nominal state vector is then used to compute both the approximate incidence matrix $\bm \hat{\bm{B}}$ and the physical-based method weights, i.e., $\bm \Omega^{AW}$ and $ \bm \Phi^{AW}$, enabling the calculation of the Laplacian matrix $\bm L^{AW}$. Finally, the problem in (\ref{eq:8}) is solved, 
obtaining the desired interpolated residual vector $\bm\Delta\hat{\bm\psi}$.


\subsubsection{Selection of virtual sensors}\label{subsubsection:select-VS}


While more information should improve the leak localization performance, the disturbance induced in the learning phase by the interpolation errors becomes, at some point, significant. Thus, we propose to select a set $\mathcal{Z}_v$ of $n_{vs}$ nodes, which will append their interpolated state to that of the actual sensors. In this way, the total set of sensors would be composed of $n_{ts}=n_{\zeta} + n_{vs}$ nodes.
\textcolor{black}{Please note that} \hlight{C4 - R4}{increasing the number of VSs means increasing $n_{ts}$ which,
as per Remark~\ref{remark:DL_complexity},
results in
a linear increase in the FDI step
and a quadratic increase in the learning process.}

This selection is performed by means of a data-driven sensor placement methodology \cite{RomeroBen2022b}, in order to minimize the use of the hydraulic model. This approach pursues the sensor configuration that minimizes a certain topological-based metric (e.g. Section 2.3 and the explanations around (6) in \cite{RomeroBen2022b}), related with the distance from the sensors to the rest of the nodes. Moreover, it allows to settle a set of pre-defined sensors, so that the rest of them are placed considering that the pre-defined ones are fixed, enabling virtual sensors to be selected in areas where the real sensor density is low. 
In this way, we can achieve the final training and testing datasets $\bm\Delta\hat{\bm\Psi}^{train}_{ts}\in \mathbb{R}^{n_{ts}\times n_{train}}$ and $\bm\Delta\hat{\bm\Psi}^{test}_{ts}\in \mathbb{R}^{n_{ts}\times n_{test}}$.

\subsubsection{Learning/classification stages}\label{subsubsection:learning-class-stages}

Once the input datasets to the learning stage are ready, the DL algorithm is trained and tested. Note that each entry of the datasets must be associated to its corresponding label, i.e., the leak scenario that originated the stored data. These matrices are referred to as $\bm H^{train}\in\R^{n_{class}\times n_{train}}$ and $\bm H^{test}\in\R^{n_{class}\times n_{test}}$. The training and testing methodologies, describing the complete AW-GSI and DL processes, are detailed in Algorithms \ref{alg:awgsi-DL-train}-\ref{alg:awgsi-DL-app}.

\begin{algorithm}[t]
\caption{AW-GSI-DL | Training}
\label{alg:awgsi-DL-train}
\begin{algorithmic}[1]
\REQUIRE $\hat{\bm\Psi}_{\zeta}^{train}, \hat{\bar{\bm\Psi}}_{\zeta}^{train}\in \mathbb{R}^{n_{\zeta}\times n_{train}}, \bm \sigma \in \mathbb{R}^{m}, \bm H^{train}\in\R^{n_{class}\times n_{train}}, s \in \R, n_{vs} \in \R$
\STATE Compute $\bm\Delta\hat{\bm\Psi}^{train}$ from $\hat{\bm\Psi}_{\zeta}^{train}, \hat{\bar{\bm\Psi}}_{\zeta}^{train}, \bm\sigma$ using Algorithm \ref{alg:aw-gsi}
\STATE Select a set $\mathcal{Z}_v$ of $n_{vs}$ virtual sensors using \cite{RomeroBen2022b}
\STATE Extract $\bm\Delta\hat{\bm\Psi}^{train}_{ts}$ from $\bm\Delta\hat{\bm\Psi}^{train}$ using $\mathcal{Z}_v$
\STATE Compute $\bm D, \bm W, \bm A$ from $\bm\Delta\hat{\bm\Psi}^{train}_{ts}, \bm H^{train}$ solving (\ref{dict_learn}) 
\RETURN $\bm D, \bm W, \mathcal{Z}_v$
\end{algorithmic}
\end{algorithm}

\begin{algorithm}[t]
\caption{AW-GSI-DL | Application}
\label{alg:awgsi-DL-app}
\begin{algorithmic}[1]
\REQUIRE $\hat{\bm\Psi}_{\zeta}^{test}, \hat{\bar{\bm\Psi}}_{\zeta}^{test}\in \mathbb{R}^{n_{\zeta}\times n_{test}}, \bm \sigma \in \mathbb{R}^{m}, s \in \R, \mathcal{Z}_v$ 
\STATE Compute $\bm\Delta\hat{\bm\Psi}^{test}$ from $\hat{\bm\Psi}_{\zeta}^{test}, \hat{\bar{\bm\Psi}}_{\zeta}^{test}, \bm\sigma$ using Algorithm \ref{alg:aw-gsi}
\STATE Extract $\bm\Delta\hat{\bm\Psi}^{test}_{ts}$ from $\bm\Delta\hat{\bm\Psi}^{test}$ using $\mathcal{Z}_v$
\STATE Compute $\bm X$ from $\bm\Delta\hat{\bm\Psi}^{test}_{ts}, \bm D, s$ using OMP
\RETURN $i = \arg\max_{j=1:n_{class}}(\bm{Wx})_j,\: \forall \bm x\in \bm X$
\end{algorithmic}
\end{algorithm}

Note that the application algorithm can be operated in real-time, providing single head state vectors (for the sensorized nodes) instead of a testing dataset matrix. The nominal information can be retrieved from a historical dataset, or from the instants before the leak detection.





\section{Results and discussion\protect\footnotemark}\label{section:results}

\footnotetext{\hlight{C1 - R1}{Code and data available at \url{https://github.com/pirofti/AW-GSI-DL}}}

The proposed methodology, AW-GSI-DL, is implemented and tested over a case study network, comparing its performance in terms of interpolation accuracy and localization results with respect to GSI-DL \cite{Irofti2022}. Moreover, both methods are compared with a purely data-driven approach, GSI-LCSM \cite{RomeroBen2022}. Thus, the advantages of the presented updates to the interpolation process are highlighted, as well as the improvements of including the learning stage to help the interpolation scheme with the localization task. 
All the localization simulations were performed on
an AMD Ryzen Threadripper PRO 3955WX with
512GB of memory using Octave 8.1.0.

\subsection{Case study}\label{subsection:case-study}

The network of Modena, Italy, constitutes a well-known case study in the WDNs management field \cite{Bragalli2012,Quinones2021, Alves2022}. It is an open-source benchmark, which helps with the replicability and comparison of results. The network topology is schematically represented in Fig.~\ref{fig:modena_scheme}, and its main properties are summarized by Table \ref{table:modena-characteristics}. Both the network size (in terms of number of nodes/pipes, reservoirs and pipe lengths) and the total demand \textcolor{black}{correspond to a problem of realistic size.}

A set of 20 pressure sensors are considered to be deployed over the network, i.e., 4 sensors at the water inlets and 16 at inner nodes. The sensor locations are decided by means of the previously mentioned data-driven sensor placement technique, presented in \cite{RomeroBen2022b}. 

\begin{table}[t]
\renewcommand{\arraystretch}{1.3}
\caption{Modena WDN characteristics}
\label{table:modena-characteristics}
\centering
\begin{tabular}{|c|c|}
\hline
\textbf{Property} & \textbf{Value}\\
\hline
Number of inner junctions & 268 \\
\hline
Number of pipes & 317 \\
\hline
Number of reservoirs (water inlets) & 4 \\
\hline
Total pipe length & 71.8 km \\
\hline
Total nodal demand & $\sim$400 l/s \\
\hline
\end{tabular}
\end{table}

\begin{figure}[htb]
\centering
\includegraphics[width=\linewidth]{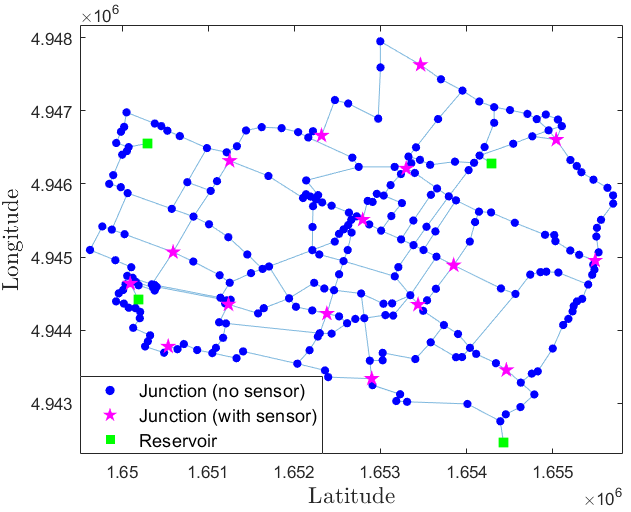}
\caption{Schematic representation of the Modena WDN.}
\label{fig:modena_scheme}
\end{figure}

\subsection{Data generation}\label{subsection:data-gen}

Hydraulic data from a wide set of leak scenarios is required to assess the performance of the methodology. To this end, EPANET has been used to generate leak data for all the possible scenarios regarding the leak location. The values of the data generation parameters presented in Section \ref{subsubsection:data_generation} in the performed EPANET simulations are the following:

\begin{itemize}
    \item All the possible leaks are considered, i.e., $n_{class} = 268$ (excluding the 4 water inlets). 
    \item The simulations cover a period of 24 hours, with a timestep of 1 hour, and hence $n_t = 24$.
    \item Several leak sizes have been considered to replicate realistic scenarios. 
    In this case, training sets have been derived considering four different leak sizes, ranging from 4 l/s to 7 l/s (i.e., $\sim$1-1.75\% of the average total inflow) with a step of 1 l/s. Three testing leak sizes have been selected to be different to the training ones 
    but within the training range: 4.5, 5.5 and 6.5 l/s. Thus, $n_{ls}=7$. \hlight{C2 - R4}{Using leaks with smaller magnitudes for training makes the FDI mechanism more challenging since the discrimination between ``healthy but affected by uncertainties'' and ``under fault'' becomes harder to check. Thus, we consider the current leak range to be adequate for both training and subsequent testing. Not least, recent studies dealing with the same benchmark consider similar or even larger leak sizes because of this fact \cite{Alves2022}}.
\end{itemize}

Additionally, several sources of uncertainty have been considered during the data generation to test the approach in realistic conditions. On the one hand, sensor precision is considered, trusting the generated head values until a precision of $\pm1$ cm. Thus, leak effects are not noticeable until the minimum level of precision is reached.
On the other hand, 
uniformly random uncertainty has been added to physical properties of the WDN, specifically pipe diameters and roughness, as well as the daily demand patterns. Two cases have been studied: 0.5\% and 1\% of uncertainty with respect to the noise-free value. Note that the methodology requires a leak-free scenario, which is also affected by uncertainty, to achieve the pressure residuals. \hlight{C5 - R4}{Thus, the input data to the learning stage is affected by a higher degree of uncertainty}. 


\textcolor{black}{Let us remark that considering all the mentioned sources of uncertainty, i.e., leak size, modeling (network) properties, sensor precision and demand pattern, is not common among most of the state-of-the-art mixed model-based/data-driven leak detection/localization methods, as it can be appreciated in Table I of \cite{RomeroBen2023}, which summarizes the features of 54 relevant articles in the field.
}

Finally, let us remark that leaks are generated by means of the emitter component available in EPANET. It models the flow through an orifice that leaks to the environment 
as

\begin{equation}
    f^{leak}_i = \epsilon \rho_i^{0.5},
\end{equation}

\noindent where $\epsilon$ is the emitter coefficient, and $\rho_i$ is the pressure at node $\mathscr{v}_i$. 

\subsection{Interpolation performance}\label{subsection:interp-perf}


In order to analyze the improvements of AW-GSI with respect to GSI in terms of state estimation performance, let us derive two different metrics related to the difference between interpolated and actual values. Specifically, they are based on the root-squared-mean error (RMSE), which is computed as

\begin{equation}
    RMSE(\bm x, \hat{\bm x}) = \sqrt{\frac{1}{n} \sum_{i=1}^{n} (x_i - \hat{x}_i)^2},
\end{equation}

\noindent where $\bm{x}$ is a generic vector whose $i-th$ element corresponds to a value at node $\mathscr{v}_i$, and $\hat{\bm x}$ is the computed estimation of $\bm{x}$.

The first metric assesses the accuracy of the hydraulic head reconstruction by calculating the error between actual EPANET and interpolated values., i.e., $RMSE(\bm\psi,\hat{\bm\psi})$. The second metric, computed as $RMSE(\bm{\Delta\psi},\bm\Delta\hat{\bm\psi})$, gauges the similarity between the residuals received by the learning scheme from interpolation and EPANET, indicating the quality of the training data. Both calculations are performed for each leak case and time instant, averaging over the latter to obtain a RMSE value for each leak.


Note again that GSI retrieves heads and AW-GSI yields residuals. Thus, for the sake of the comparisons, it is required to compute the residuals associated to the GSI heads and the heads that correspond to the AW-GSI residuals. For GSI, the difference between the each leak head vector and the corresponding leak-free data yields the GSI associated residuals. For AW-GSI, the nominal heads vector is added to the estimated residuals to obtain the leak hydraulic heads.

The interpolation results are displayed in Fig.~\ref{fig:residual_and_head_errors}. The figure is divided into three columns and two rows. Each column represents a different scenario regarding the considered uncertainty. Specifically, the cases of 0, 0.5 and 1\% of uncertainty are exposed. Regarding the rows, the first one shows the results for the residual-error, whereas the second row shows the head-reconstruction error. Each vertical line denotes the RMSE of the corresponding leak, with blue and red lines showing GSI and AW-GSI performance respectively. Note that the mean values (over the leaks) are indicated with cyan and green horizontal lines for GSI and AW-GSI respectively. Important conclusions can be extracted from these results:


\begin{figure}[t]
\centering
\includegraphics[width=\linewidth]{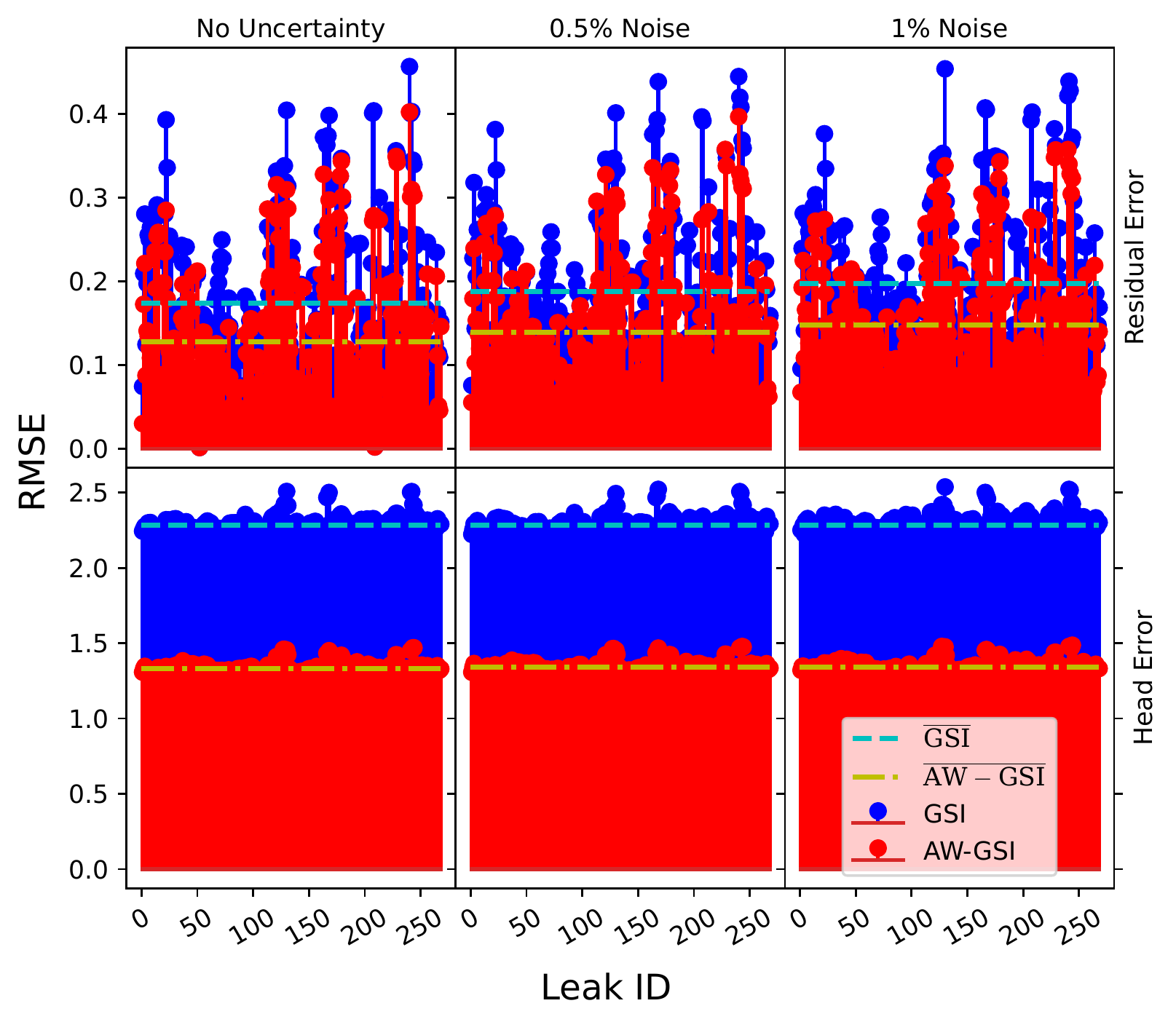}%
\caption{Interpolation RMSE per leak -- GSI vs AW-GSI -- head and residual errors.}
\label{fig:residual_and_head_errors}
\end{figure}




\begin{itemize}
    \item About the head-reconstruction error, AW-GSI clearly outperforms GSI, yielding a lower error for all the possible leaks. An analysis of the mean values (over the leaks) reveals a reduction of 41.65\% in the head-error. Thus, we conclude that AW-GSI is clearly capable of estimating the network hydraulic heads with a higher degree of accuracy.
    \item Regarding the residual-error, the difference in terms of performance between both interpolation schemes is not as large as in the previous case. Even so, the error associated to AW-GSI is lower than the error of GSI in a 88.06\% of the leak scenarios, with an average (over the leaks) reduction of 26.62\% in RMSE. Thus, AW-GSI arguably performs better than GSI regarding residual-error too. 
    \item Analyzing the effect of uncertainty, it becomes clear that the interpolation process is not degraded. Both GSI and AW-GSI are robust against physical properties and demand pattern noise. This is specially interesting for AW-GSI, considering that the AW-generation is fed with the noise-free pipe diameters and roughness, and hence the weights are obtained with information that slightly differs from the one used to generate the hydraulic data.
\end{itemize}

\subsection{Localization performance}\label{subsection:loc-performance}

The interpolation results from the previous stage, for both GSI and AW-GSI, were used as inputs to the DL process, leading to the training and posterior testing of the methodologies for leak localization. Note that DL is fed with normalized residuals, so they must be computed in the case of GSI from the interpolated leak and leak-free vectors (both GSI and AW-GSI residuals have to be normalized).

The localization results are presented at Fig.~\ref{fig:loc_performance}. It is divided into three columns, each displaying performance for different uncertainty levels (0\%, 0.5\%, and 1\%). In each column, three accuracy levels are compared based on the number of successful targets considered: 

\begin{itemize}
    \item Node-level accuracy (indicated by circle markers): indicates the percentage of testing samples for which the exact leaky node is selected by the trained DL scheme.
    \item 1-neighbour level accuracy (indicated by triangle markers):  denotes the percentage of successful tests when the 1st layer of neighbours of the exact leaky node are considered as correct locations too.
    \item 2-neighbour level accuracy (indicated by square markers):  it is analogue to the previous metric, but considering the exact node and the first and second layers of neighbours.
\end{itemize} 

Besides, the x-axis stands for the number of virtual sensors that are considered in each case. In this way, an analysis of the effect of including extra VS can be performed. Specifically, the cases of 0, 10, 20, 50, 100, 150, 200 and 252 are included (this last case implies that all the non-sensorized nodes provide information to the learning process through the interpolation data). Note that the previously mentioned sensor placement methodology \cite{RomeroBen2022b} is used to add a set of new sensors (virtual) to the existing set of installed ones. In this way, the VS are located at areas with a low density of pressure sensors.

\begin{figure}[t]
\centering
\includegraphics[width=\linewidth]{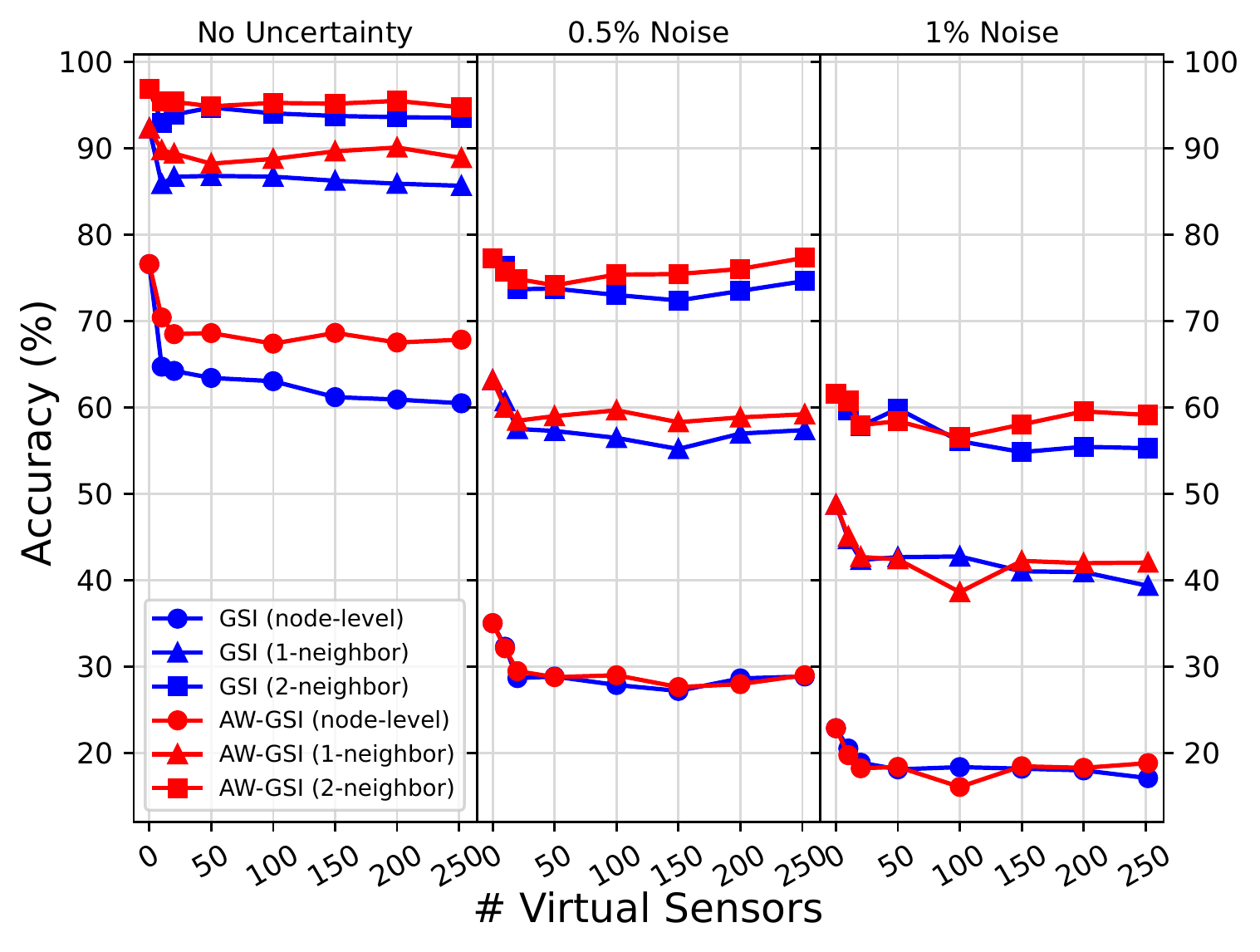}%
\caption{Localization performance $-$ GSI vs AW-GSI. The x-axis shows the number of VS considered, whereas the y-axis shows the localization accuracy (\%). The three columns of subfigures represent different uncertainty levels: no uncertainty, 0.5\% and 1\%.}
\label{fig:loc_performance}
\end{figure}


\begin{table}[t]
\caption{Execution times for the localization results presented in Fig.~\ref{fig:loc_performance}.}
\label{table:times}
\footnotesize
\centering
\begin{tabular}{c|c|c|c|c|c}
\textbf{\#VS} & 0 & 50 & 100 & 150 & 252\\
\hline
\textbf{t(min)} & $73\pm 0.7$ & $74\pm 0.9$ & 
$75\pm 1$ & $76\pm 0.8$ & $77\pm 0.7$\\
\end{tabular}
\end{table}

These results can be analyzed from various perspectives:

\begin{itemize}
	\item Regarding the comparison between both interpolation methodologies, it is clear that the localization accuracy is greater when AW-GSI is used instead of GSI
 , for the vast majority of tested scenarios. 
	\item The inclusion of uncertainty degrades the localization performance, affecting both interpolation schemes. Nevertheless, the accuracy results for both 1 and 2-level neighbors are mostly satisfactory, considering that, in average, the 1-level layer of neighbors only occupies a 1.25\% of the total area of network (considering the proportion of neighbors by total WDN nodes), whereas the 2-level layer of neighbors occupies 2.55\% of the area.
	\item Note that the 2-neighbors level accuracy is closer between GSI and AW-GSI when no uncertainty is considered, while greater differences occur for the node-level case. However, the opposite is remarked when uncertainty is introduced. This occurs because uncertainty mostly hinders node-level accuracy, making node-level successful localizations with no uncertainty to fall to the category of 1 or 2-neighbors level successful localizations.
	\item Regarding the number of VSs, it can be observed that the accuracy drops in comparison to the no-VS case. This is caused by the interpolation still introducing defects in the input data to the interpolation process. \textcolor{black}{When starting to introduce VSs,} \hlight{C3 - R4}{a sudden drop in the accuracy occurs, mostly in the node-level, due to the mismatch}\hlight{C5 - R4}{ that may occur between measured and interpolated data. This drop is diminished or even overcome, as the extra information introduced when including more and more VS is consistent with the already provided virtual estimations}. However, let us remark that AW-GSI outperforms GSI in terms of diminishing the performance drop, even achieving higher performances for 252 VS than for more reduced sets of VS in some scenarios.
\end{itemize}

To complement the localization performance results,
we also present in Table~\ref{table:times} the average execution times of the training process
as the number of virtual sensors increase.
We observe a slight increase as the number grows,
but it is not significant or in any way prohibiting
the use of numerous virtual sensors.
\hlight{C4 - R4}{Profiling showed that most time is spent during the sparse representation stage (79\%),
followed by the dictionary update (19\%)
and the rest is spent iterating over the two.}

Another experiment has been carried out to highlight the contribution of the combination of interpolation and learning to increase the localization accuracy. Both interpolation-learning strategies, i.e., GSI-DL and AW-GSI-DL, are compared with the localization strategy, based on a leak candidate selection method (LCSM), which originally exploited GSI to interpolate approximated hydraulic heads: GSI-LCSM \cite{RomeroBen2022}. 

The graphical result of this comparison is presented in Fig.~\ref{fig:comparison_GSILCSM_GSIDL_AWGSIDL}. In this figure, the x-axis stands for the considered depth of neighbours of the actual leak that are considered as successful targets for the classification process (as detailed for the cases of node-level, 1-neighbour level and 2-neighbour level accuracy in the explanations of Fig.~\ref{fig:loc_performance}). This depth is referred to as the localization area, as it correspond to the set or area of nodes that constitute a correct localization. Thus, the aim of this experiment is to study the evolution of the accuracy from node-level localization up to a certain degree of neighbours. In this case, a maximum of 6-neighbour level is considered, which represents a 13.31\% of the total number of network nodes (in average). This selection is justified by the original nature of GSI-LCSM: the candidate selection stage of this method dynamically selects the size of the set of candidates (which can be regarded as the set of successful targets), as the methodology was designed to yield localization areas. For the sake of the comparison between GSI-LCSM and the two DL-based strategies, this dynamic set selection was substituted by the predefined sets of neihgbours corresponding to each node and depth level. Thus, we considered appropriate to reach at least a 12-13\% of the network with the maximum level, regarding the results obtained in terms of candidate areas in previous works based on GSI-LCSM \cite{Romero2021}. Regarding GSI-DL and AW-GSI-DL, let us remark that the case of 252 VS (that is, all the network nodes are real or virtual sensors) is selected for the comparison, considering that GSI-LCSM exploits all the interpolated information during LCSM.

The results presented yield several interesting conclusions:

\begin{itemize}
    \item The DL-based methods greatly improve accuracy at node-level (and reduced depths of neighbors) in comparison to GSI-LCSM. This confirms the suitability of the approach to accomplish one of the aims we were pursuing, specifically the improvement of the usually unreliable localization of purely data-driven methods at node-level.
    \item Although the performance of the three methods become more similar for greater depths of neighbors, the DL-based approaches continues yielding higher accuracy.
    \item Noise degrades the node-level (and reduced depths of neighbors) accuracy of the learning-based techniques, although the performance still improves GSI-LCSM, whose results are not affected by noise. 
\end{itemize}


\begin{figure}[t]
\centering
\includegraphics[width=\linewidth]{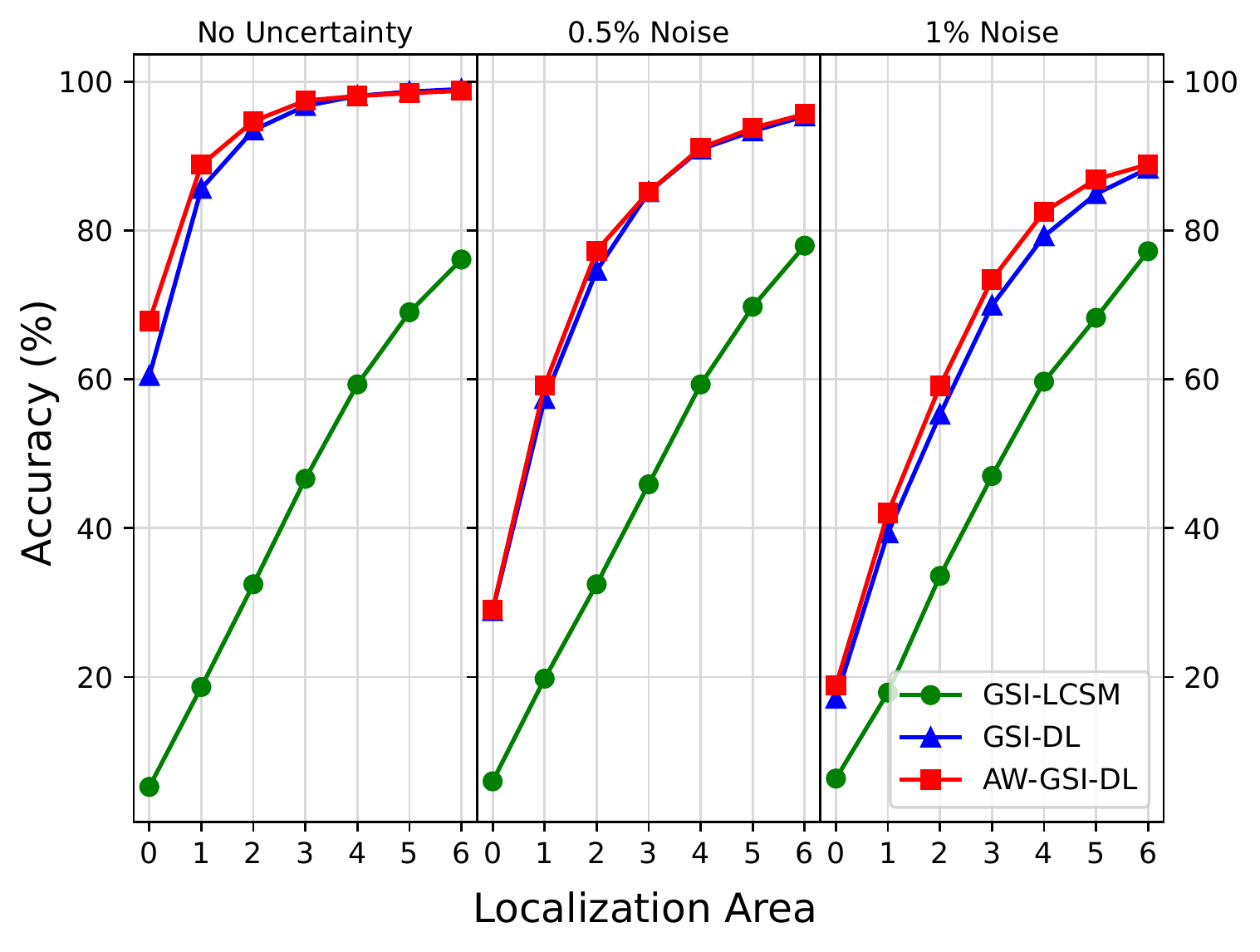}
\caption{Localization performance $-$ GSI-LCSM vs GSI-DL vs AW-GSI-DL. The x-axis denotes the depth of neighbors of the leaky node considered as correct targets, and the y-axis represent the localization accuracy (\%). From top to bottom, the subplots represent the cases considering 0\%, 0.5\% and 1\% noise respectively.}
\label{fig:comparison_GSILCSM_GSIDL_AWGSIDL}
\end{figure}

\section{Conclusions}\label{section:conclusions}

This article presents a leak localization method combining interpolation and dictionary learning. The novel interpolation technique, i.e., AW-GSI, leverages the physics behind the relation between hydraulic heads of neighboring nodes in WDNs. Additionally, residuals are directly interpolated instead of hydraulic head values. 

The method is applied to the well-known Modena case study, proving AW-GSI's improvements over GSI in terms of interpolation error, considering state and residual estimation; and posterior localization. Performance is tested for several uncertainty levels to assess its effect over the estimated heads/residuals and localization results. We note that uncertainty does not affect the interpolation process, although it degrades the localization accuracy up to some extent. Besides, localization experiments involved multiple virtual sensor sets and various degrees of target node depth. 

Improvements are necessary to enhance the interpolation process and enrich the information for the learning stage. Additionally, boosting dictionary learning's robustness to noisy data can mitigate classification degradation caused by uncertainty. \hlight{C5 - R4}{This will allow to overcome the limitations of the learning approaches caused by the noise, leading to a useful tool for water utilities to work with real data}. Nevertheless, promising results have been obtained in this work, showing improvements with respect to previous methods, and confirming the suitability of the combined interpolation-learning scheme to compensate the weaknesses of interpolation-only methods.

\appendix[Proof of Lemma~\ref{lem:implicit}]

The implicit function theorem states that for a function $\mathbf f:\mathbb R^{n+m}\mapsto \mathbb R^m$ which respects the implicit relation $\mathbf f(\mathbf x,\mathbf y)=\mathbf 0$, with $\mathbf x\in \mathbb R^m$, $\mathbf y\in \mathbb R^n$, in a neighborhood around $(\bar{\mathbf x}, \bar{\mathbf y})$ verifying $\mathbf f(\bar{\mathbf x}, \bar{\mathbf y})=0$, there exists \hlight{C2 - R11}{(under reasonable assumptions, \cite{protter1985implicit})} an explicit dependence $\mathbf y=\mathbf g(\mathbf x)$ such that $\mathbf f(\mathbf x,\mathbf g(\mathbf x))=\mathbf 0$. Moreover, we have that
\begin{equation}
\label{eq:implicit_function}
    \frac{\partial \mathbf g(\mathbf x)}{\partial \mathbf x_j}=-\left[\frac{\partial \mathbf f(\mathbf x,\mathbf y)}{\partial \mathbf y}\right]^{-1}\cdot \frac{\partial \mathbf f(\mathbf x,\mathbf y)}{\partial \mathbf x_j},\: \forall j=1\dots n.
\end{equation}

Equation \eqref{eq:wdn_hw} defines an implicit relation between the current junction's head $\psi_i$ and its neighboring junctions' heads $\{\psi_j\}_{j:\:  b_{ij}\neq 0}$. We may then define the implicit function $\mathbf f_i(\psi_i, \bm{\psi}_{\J})=0$ where $\mathbf f_i:\mathbb R^{|\J|+1}\mapsto \mathbb R$ and which comes from moving $\mathbf c_i$ to the left of the equal sign in \eqref{eq:wdn_hw} and by grouping all neighboring node's head values in the vector $\bm{\psi}_{\J}$, defined as below \eqref{eq:wdn_g}. Note that the construction is particular in the sense that $\mathbf f_i$ is a scalar function and thus the Jacobian inversion from \eqref{eq:implicit_function} is a simple scalar division.

\noindent Applying the apparatus of the implicit function theorem we have that for a pair $(\bar{\psi}_i, \bar{\psi}_{\J})$ verifying $\mathbf f_i(\bar{\psi}_i,\bar{\psi}_{\J})=0$ there exists the scalar function $\mathbf g_i:\mathbb R^{|\J|}\mapsto \mathbb R$ which verifies $\psi_i=\mathbf g_i(\bm{\psi}_{\J})$ on a neighborhood of $(\bar{\psi}_i, \bar{\psi}_{\J})$, thus proving \eqref{eq:wdn_g} and allowing to adapt \eqref{eq:implicit_function} into
\begin{equation}
    \label{eq:implicit_function_hi}
    \frac{\partial \mathbf f_i(\bm{\psi}_{\J})}{\partial \psi_j}=-\left[\frac{\partial \mathbf f_i(\psi_i,\bm{\psi}_{\J})}{\partial \psi_i}\right]^{-1}\cdot \frac{\partial \mathbf f_i(\psi_i,\bm{\psi}_{\J})}{\partial \psi_j},\: \forall j\in \J.
\end{equation}
Differentiating after $\psi_i$ in \eqref{eq:wdn_hw}, we obtain
\begin{align}
    \nonumber\frac{\partial \mathbf f_i(\psi_i,\bm{\psi}_{\J})}{\partial \psi_i}&=\frac{\partial}{\partial \psi_i}\mkern-8mu\left(\sum\limits_{k:\:  b_{ik}\neq 0} \mkern-16mu b_{ik}\sigma_{(i,j)}^{0.54}\left[ b_{ik}(\psi_i-\psi_k)\right]^{0.54}\mkern-16mu-\mathbf c_i\right),\\
    \label{eq:hw_current_node_i}&=0.54\mkern-16mu\sum\limits_{k:\:  b_{ik}\neq 0} \mkern-16mu b^2_{ik}\sigma_{(i,j)}^{0.54}\left[ b_{ik}(\psi_i-\psi_k)\right]^{-0.46}\mkern-36mu.
\end{align}
We repeat the procedure for each $\psi_j, j\in \J$ and obtain
\begin{equation}
\label{eq:hw_current_node_j}\frac{\partial \mathbf f_i(\psi_i,\bm{\psi}_{\J})}{\partial \psi_j}=-0.54 b^2_{ij}\sigma_{(i,j)}^{0.54}\left[ b_{ij}(\psi_i-\psi_j)\right]^{-0.46}.
\end{equation}
Noting that $ b_{ij}^2=1,\forall  b_{ij}\neq 0$, we introduce \eqref{eq:hw_current_node_i} and \eqref{eq:hw_current_node_j} into \eqref{eq:implicit_function_hi}, thus arriving at \eqref{eq:wdn_gderiv} and concluding the proof.


%



\section*{Acknowledgment}

The authors would like to thank the Spanish national project L-BEST (Ref. PID2020-115905RB-C21).
This work was also supported by a grant of the Ministry of Research, Innovation and Digitization, CCCDI - UEFISCDI, project number PN-III-P2-2.1-PED-2021-1626, within PNCDI III.

\ifCLASSOPTIONcaptionsoff
  \newpage
\fi



%



\bibliographystyle{IEEEbib}
\bibliography{gsi-dl} 

%

\begin{IEEEbiography}[{\includegraphics[width=1in,height=1.25in,clip,keepaspectratio]{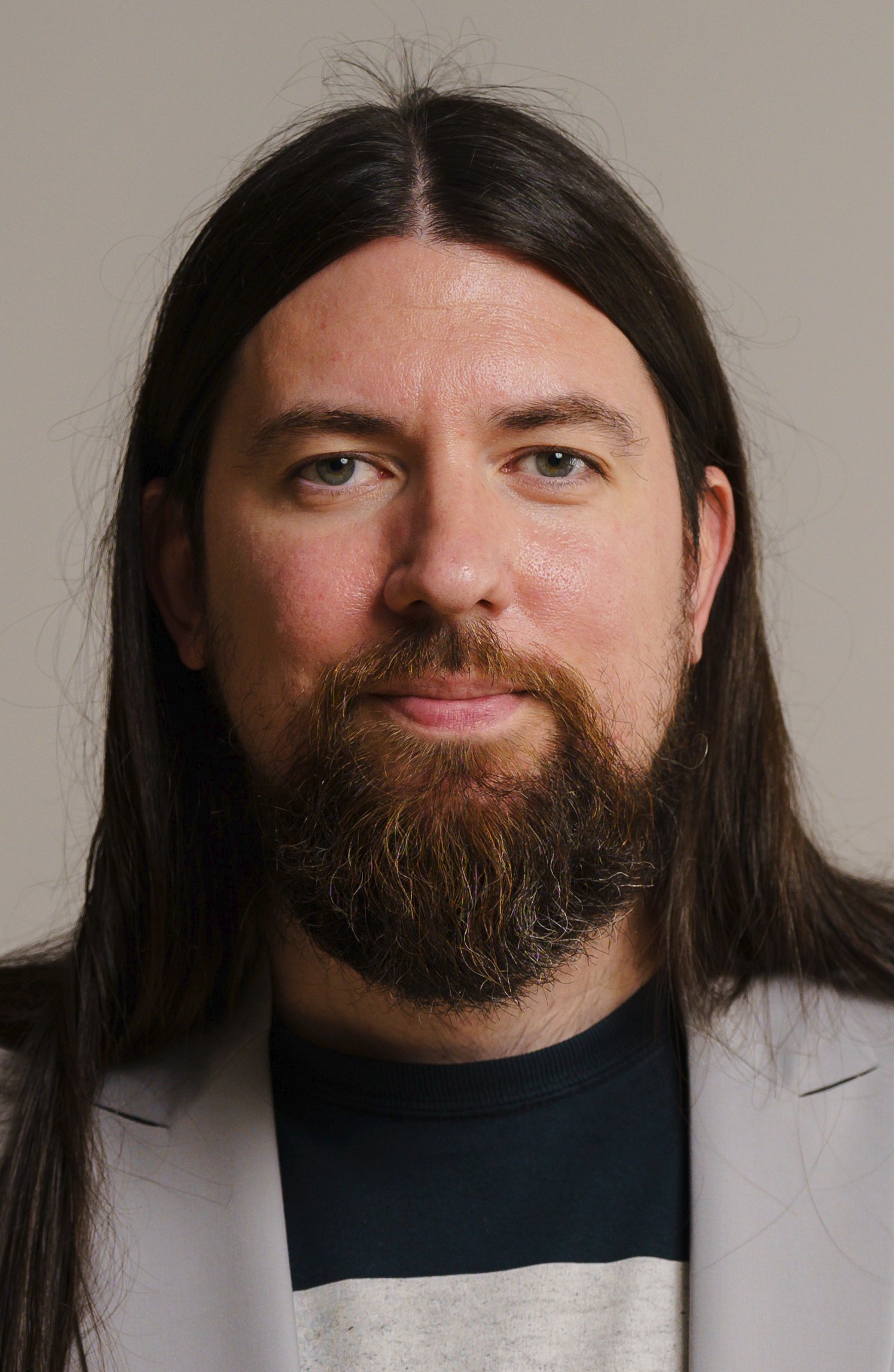}}]{Paul Irofti}
is an Associate Professor
within the Computer Science Department
of the Faculty of Mathematics and Computer Science
at the University of Bucharest
and the Vice-President
of the
Institute for Logic and Data Science.
He is the co-author of the book “Dictionary Learning Algorithms and Applications” (Springer 2018)
awarded by the Romanian Academy.
He is PhD in Systems Engineering at the Politehnica University of
Bucharest since 2016.
His interests are anomaly detection, signal processing, 
numerical algorithms and optimization.

\end{IEEEbiography}

\begin{IEEEbiography}[{\includegraphics[width=1in,height=1.25in,clip,keepaspectratio]{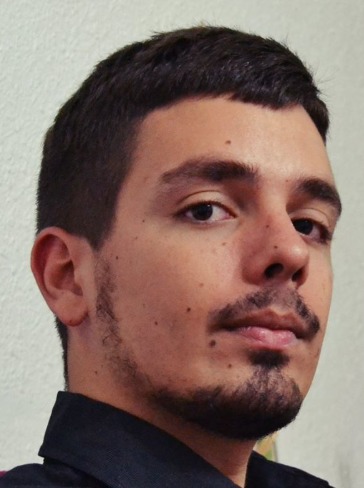}}]{Luis Romero-Ben}
is currently a PhD student and researcher at the Automatic Control group at the Institut de Robòtica e Informàtica Industrial (CSIC-UPC), Spain. His interests are focused on the development and application of data-driven methodologies for the control and monitoring of water distribution networks and urban drainage systems.
\end{IEEEbiography}


\begin{IEEEbiography}[{\includegraphics[width=1in,height=1.25in,clip,keepaspectratio]{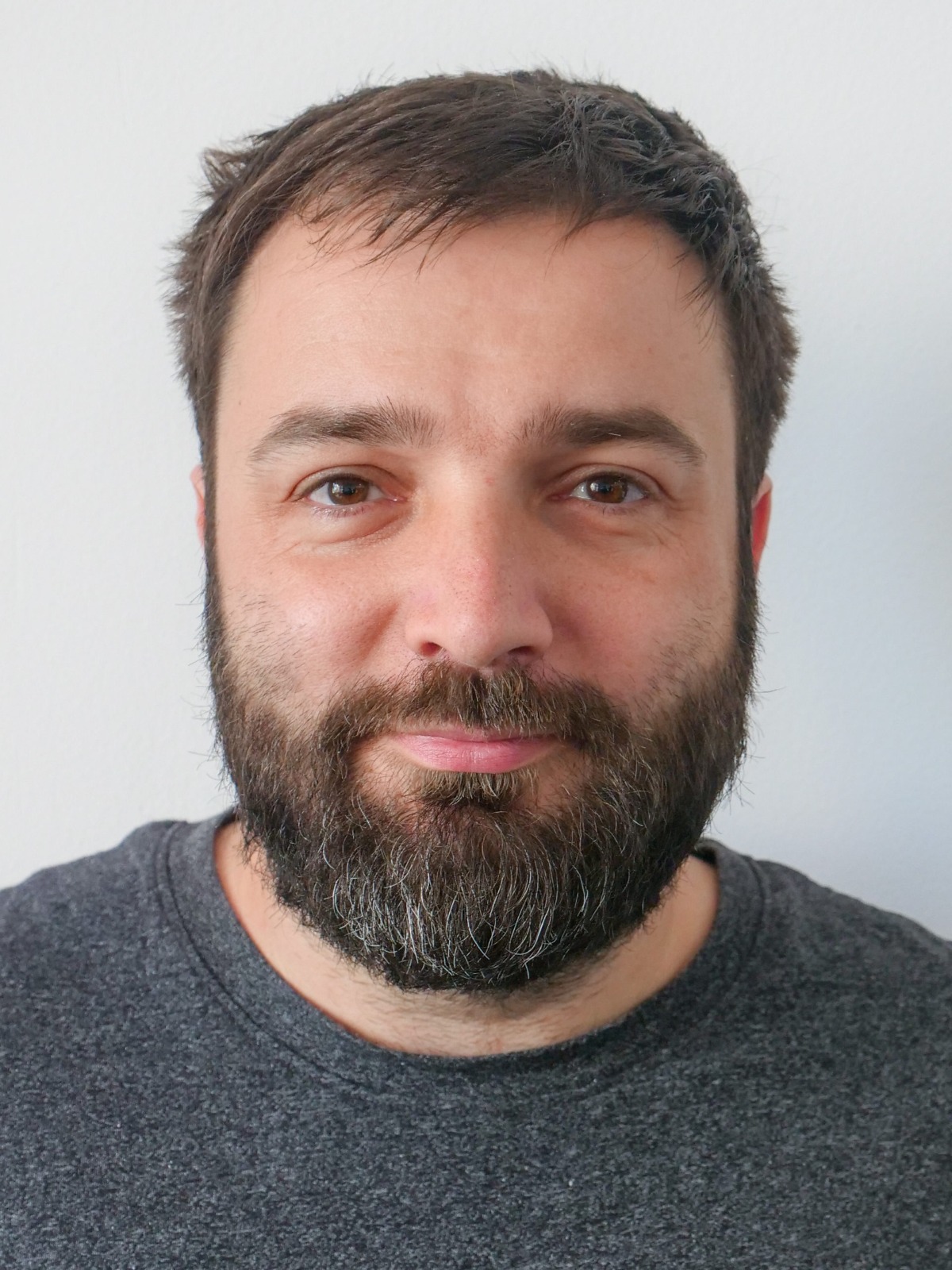}}]{Florin Stoican}
 is Professor in the department of Automatic Control and Systems Engineering, Politehnica University of Bucharest. He obtained his PhD in Control Engineering in 2011 from Supelec (now CentraleSupelec), France with an application of set-theoretic methods for fault detection and isolation. His interests are constrained optimization control, set theoretic methods, fault tolerant control, mixed integer programming, motion planning.
\end{IEEEbiography}

\begin{IEEEbiography}[{\includegraphics[width=1in,height=1.25in,clip,keepaspectratio]{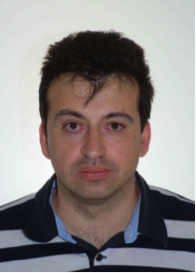}}]{Vicenç Puig}
 holds a PhD in Automatic Control, Vision and Robotics and is the leader of the research group Advanced Control Systems (SAC) at the Polytechnic University of Catalonia.
 He has important scientific contributions in the areas of fault diagnosis and fault tolerant control using interval models. He participated in more than 20 international and national research projects in the last decade. He led many private contracts, and published more than 80 articles in JCR journals and more than 350 in international conference/workshop proceedings. Prof. Puig supervised over 20 PhD theses and over 50 MA/BA theses.
\end{IEEEbiography}


\vfill


\end{document}